\pdfoutput=1
\newif\ifFull
\Fullfalse

\documentclass{llncs}
\usepackage{graphicx}
\usepackage{url}
\usepackage{epstopdf}
\usepackage{subfigure}
\usepackage{boxedminipage}
\usepackage{amsmath}
\usepackage{amssymb}
\usepackage{algorithmic}
\usepackage{color}

\usepackage[lined,boxed]{algorithm2e}
\usepackage[utf8]{inputenc}
\usepackage{hyperref}

\makeatletter
\def\@begintheorem#1#2{\sl \trivlist \item[\hskip \labelsep{\bf #1\ #2:}]}
\def\@opargbegintheorem#1#2#3{\sl \trivlist
      \item[\hskip \labelsep{\bf #1\ #2\ #3:}]}
\makeatother

\begin{document}

%
%

\title{Simple Multi-Party Set Reconciliation}
\author{Michael Mitzenmacher\inst{1}\thanks{Supported in part by NSF grants CCF-1320231, CNS-1228598, IIS-0964473, and CCF-0915922.} \and Rasmus Pagh\inst{2}\thanks{Supported by the Danish National Research Foundation under the Sapere Aude program.}}
\institute{Harvard University, School of Engineering and Applied Sciences, \email{michaelm@eecs.harvard.edu}
 \and IT University of Copenhagen, \email{pagh@itu.dk}}
\date{}

\maketitle

\begin{abstract}
Many distributed
cloud-based services use multiple loosely consistent replicas of user
information to avoid the high overhead of more tightly coupled
synchronization.  
Periodically, the information must be
synchronized, or reconciled.  One can place this problem in the
theoretical framework of {\em set reconciliation}: two parties $A_1$
and $A_2$ each hold a set of keys, named $S_1$ and $S_2$ respectively,
and the goal is for both parties to
obtain $S_1 \cup S_2$.  Typically, set reconciliation is interesting
algorithmically when sets are large but the set difference
$|S_1-S_2|+|S_2-S_1|$ is small.  In this setting the focus is on
accomplishing reconciliation efficiently in terms of communication;
ideally, the communication should depend on the size of the set
difference, and not on the size of the sets.

In this paper, we extend recent approaches using Invertible Bloom
Lookup Tables (IBLTs) for set reconciliation to the multi-party
setting.  There are three or more parties
$A_1,A_2,\ldots,A_n$ holding sets of keys $S_1,S_2,\ldots,S_n$
respectively, and the goal is for all parties to obtain $\cup_i S_i$.
Whiel this could be done by pairwise reconciliations, we seek
more effective methods.  Our general approach can function even if the number
of parties is not exactly known in advance, and with some additional cost
can be used to determine which other parties hold missing keys.  

Our methodology uses network coding techniques in conjunction with
IBLTs, allowing efficiency in network utilization along with
efficiency obtained by passing messages of size $O(|\cup_i S_i - \cap_i S_i|)$.
By connecting reconciliation with network coding, we can provide 
efficient reconciliation methods for a number of natural 
distributed settings.  
\end{abstract}




\medskip

\section{Introduction} 

As users migrate information to cloud storage, the burden of
reliability moves to the cloud provider.  Thus many cloud vendors such
as Amazon \cite{dynamo} and Azure \cite{bradpaper} use multiple
loosely consistent replicas of user information because of the high
overhead of keeping replicas synchronized at all times.  Further,
users often retain copies of their information on laptops, tablets,
phones and Personal Digital Assistants (PDAs); these devices are often disconnected from cloud
storage and thus can diverge from the corresponding copies in the
cloud.  The situation naturally grows even more complicated when
multiple users have access to information, because the number of replicas can
increase with the number of users.  Periodically, however 
copies of information objects must be synchronized or {\em reconciled}.  One can also view
the need for reconciliation at a higher level, such as for loosely 
consistent replicas of large databases that may be used for availability
by information providers.

This paper focuses on the basic problem of set
reconciliation.  In the 2-party setting, two parties $A_1$ and $A_2$
respectively have (usually very similar) sets $S_1$ and $S_2$, and want to
reconcile so both have $S_1 \cup S_2$.  Our major contribution is to
extend the recent approach to set reconciliation for two parties using
Invertible Bloom Lookup Tables (IBLTs) to the multi-party setting,
where there are three or more parties holding sets
$S_1,S_2,S_3,\ldots,S_n$, and the goal is for all parties to obtain
$\cup_i S_i$.  This could of course be done by pairwise
reconciliations, but we seek more efficient methods.  We first extend
the IBLT approach, showing that in the multi-part setting we can
reconcile using messages of size $O(|\cup_i S_i - \cap_i S_i|)$.  This
generalizes results from the two-party setting, where the information
theoretic goal has been to send information close to the size of the
set difference, rather than sending information proportional to the size
of the sets, as generally the set difference is very small compared to
the set sizes.  Our approach has other advantages, including that one
does not need to know the number of parties in advance.  Our main approach
here, related to network coding, is to think of the information stored in
the IBLT as corresponding to vectors over a suitable finite field instead
of the binary vectors used in previous work.

We further show that our methodology allows using further network
coding techniques in conjunction with IBLTs, providing additional
efficiency in terms of network utilization.  By connecting
reconciliation with network coding, we can provide more efficient
reconciliation methods that apply to a number of natural distributed
computing problems.  For example, using recent results from gossip
algorithms, we show that multi-party set reconciliation over a network
with $n$ nodes can be done in $O(\phi^{-1} \log n)$ rounds of
communication with IBLTs, where $\phi$ is the conductance of the network. 

While our work can be seen as a specific example of a linear sketch
that has a natural affinity to the network coding approach, we believe
it suggests that other linear sketch-based data structures may also 
find expanded use by combining them with ideas from network coding.

\subsection{Potential Applications and Related Work}

The use of IBLTs for distributed synchronization has already been
proposed for specific applications.  For example, recently the Bitcoin
community has considered using IBLTs for scalable synchronization of
transactions\footnote{See also \url{http://www.reddit.com/r/Bitcoin/comments/2hchs0/scaling_bitcoin_gavin_begins_work_on_invertible/}
for further discussion.} \cite{github,CCN}. Multi-party variations would
be potentially useful in the Bitcoin setting, where multiple parties
may need to track transactions.

In the setting of data centers, as shown in the survey of Bailis
and Kingsbury~\cite{Bailis:2014:NR:2639988.2655736}, network errors abound 
in cost-effective large-scale environments.
Thus, even if we attempt to keep multiple copies of data synchronized,
the synchronization will periodically fail along with the network, leaving
the problem of reconciling the differences.

As related work, we note that a different generalization of the set reconciliation problems, to
settings where a certain type of \emph{approximate} reconciliation is desired, was
recently considered in the database
community~\cite{Chen:2014:RSR:2588555.2610528}. However, that work
focuses on the setting of two parties, leaving the question of scaling
to many parties open.

Another related model for problems on distributed data is that
of \emph{distributed tracking} (see
e.g.~\cite{Huang:2012:RAT:2213556.2213596}).  Our problem
differs from distributed tracking in two respects: 
we focus on exact computation (with an arbitrarily small error
probability), and we focus on periodic, as opposed to continuous,
computation of the joint function.

Several further specific applications for set difference structures are given
in \cite{setdifference}, including peer-to-peer transactions,
deduplication, partition healing, and synchronizing parallel
activations (e.g., of independent crawlers of a search engine).  They
also discuss why logging as an alternative may have disadvantages in
multiple contexts; an example they provide is for ``hot'' data items
that are written often and may therefore be in the log multiple times.
We refer the reader to this paper for more information on these
examples.  Multi-party variations of IBLT-based
synchronization methods could enhance their desirability in
these applications when multiple parties naturally arise.
In synchronizing parallel activations, for instance, several agents
in a distributed system could be gathering information into 
local databases in a redundant fashion for near-optimal accuracy in
the collection process, and then need to reconcile these local databases
into a synchronized whole.

\subsection{Background}

We briefly summarize known results for the historically common setting
of two parties with direct communication.  Consider two parties $A_1$
and $A_2$ with sets $S_1$ and $S_2$ of keys from a universe $U$.  An
important value is the size of the set difference  between $S_1$ and $S_2$, 
denoted by $d = |S_1-S_2|+|S_2-S_1| = |(S_1 \cup S_2) - (S_1 \cap S_2)|$.  In this setting
there are communication-efficient algorithms when $d$, or a good approximate upper
bound for $d$, is known.  Hence, in some algorithms for set
reconciliation, there are two phases: in a first phase a bound on $d$
is obtained, which then drives the second phase of the algorithm,
where reconciliation occurs.  See \cite{minsky2003set} for further
discussion on this point.    

One previous approach to set reconciliation uses characteristic
polynomials, in a manner reminiscent of Reed-Solomon
codes \cite{minsky2003set}.  Treating keys as values, $A_1$
considers the characteristic polynomial
$$\chi_{S_1}(Z) = \prod_{x_i \in S_1} (Z-x_i),$$
and similarly $A_2$ considers $\chi_{S_2}(Z)$.  
Observe that in the rational function
$$\frac{\chi_{S_1}(Z)}{\chi_{S_2}(Z)},$$
common terms cancel out, leaving a rational function in $Z$
where the sums of the degrees of the numerator and denominator is at most $d$.
The rational function can be determined through interpolation by evaluating the function at $d+1$ points
over a suitably large field;  hence, if $A_1$ sends the value of $\chi_{S_1}(Z)$ at $d+1$
points and $A_2$ sends the value of $\chi_{S_2}(Z)$ at the same (pre-chosen) $d+1$ points, then
the two parties can reconcile their set difference.  If the range can be embedded in a field of size
$q$, the total number of bits sent in each direction would be approximately $(d+1) \lceil \log_2 q \rceil$.   
Note that this takes $O(d^3)$ operations using standard Gaussian elimination techniques. 
Similar ideas underlie similar results by Juels and Sudan \cite{JS};  these ideas can be extended to use
other codes, such as BCH codes, with various computational tradeoffs \cite{DORS}. 

Recent methods for set reconciliation have centered on using
randomized data structures, such as the Invertible Bloom Filter or the
related but somewhat more general Invertible Bloom Lookup Table (IBLT)
\cite{IBF,setdifference,IBLT}.  For our purposes, the Invertible Bloom Lookup Table is
a randomized data structure storing a set of keys that
supports insertion and deletion operations, as well as a listing
operation that lists the keys in the structure.  We review
the use of IBLTs for 2-party set reconciliation in
Section~\ref{sec:review}.  The main effect of using IBLTs is that one
can give up a small constant factor in the data transmitted to obtain speed and
simplicity that is desirable for many implementations.  2-party reconciliation
using IBLTs can generally be done in linear time, using primarily hashing and XOR
operations.  As we show in this paper, the use of IBLTs can also be extended to 
multi-party reconciliation.

While the theory set reconciliation among two parties has been widely
studied, there appears to be no substantial prior work (that we are
aware of) specifically examining the theory of multi-party
reconciliation schemes, although the question of how to implement them
was raised in \cite{setreconc}.\footnote{Some preliminary results for
multi-party settings, also using the IBLT framework but based on
pairwise reconciliations, were provided to us by Goyal and
Varghese \cite{GV}. 
Also, after the appearance of this work on the arxiv, multi-party set reconciliation
using characteristic polynomials and repeated pairwise reconciliations was examined in \cite{BoralM};  their conclusion is that is while it is possible,
it currently seems much less efficient than the methods considered here.}
Special cases, such as rumor spreading (see,
e.g., \cite{chierichettialmost,DebMedard,demers,giakkoupis,haeupler,shah}),
where one (or more) parties have a single key to share with everyone,
have been studied, however.  

In more practical work, reconciliation among multiple parties has been
studied, often in the context of distributed data distribution, using
such techniques as erasure coding and Bloom filters to enhance
performance or reduce the overall amount of data transferred
(e.g., \cite{bcmr,bullet,lcsr,setdifference}).  However, these works
are also based on pairwise reconciliation, and some do not attempt to 
achieve data transmission proportional to the size of the set difference,
which is our goal here.  The work closest to ours is \cite{setdifference},
which also uses Invertible Bloom Filters, but is focused on pairwise reconciliation.  

\section{Review: The 2-Party Setting}
\label{sec:review}

We review 2-party set reconciliation, using the framework of the
Invertible Bloom Filter/ Invertible Bloom Lookup Table
(IBLT) \cite{IBF,setdifference,IBLT}.  More concretely, we first
describe an IBLT and its use for set reconciliation.  IBLTs store
keys, which here we will think of as fixed-length bit strings.
An IBLT is designed with respect to a threshold number of keys, $t$, so
that listing will be successful with high probability if the actual
number of keys in the structure at the time of a listing operation is
less than or equal to $t$.  
An IBLT ${\cal B}$ consists of a lookup table $T$
of $m$ cells initialized with all entries set to 0, where $m$ is
$O(t)$, and the constant factor in the order notation is generally small (between 1 and 2)
depending on the parameters chosen.
Like a standard Bloom filter, an
IBLT uses a set of $k$ random hash functions, $h_1$, $h_2$, $\ldots$,
$h_k$, to determine where keys are stored.\footnote{To
obtain structures where $m/t$ is very close to 1, one must use
irregular IBLTs, where different keys utilize a different number of
hash functions.  We simplify our description here and use regular
IBLTs, where the same number of hash functions are used for each key.
See \cite{IBLT,Rink} for more discussion.}  For simplicity, we assume
random hash functions here, and for technical reasons we assume that
the hashes of each key yield distinct locations; hence the $k$ hashes
yield a uniform subset of $k$ distinct cells from the $m \choose k$ possibilities.
Alternatively, this could easily be accomplished, for
instance, by splitting the table into $k$ subtables and having the
$i$th hash function choose a location independently and uniformly at
random in the $i$th subtable, which
would not asymptotically change the thresholds \cite{BWZ}.

The IBLT uses a hash function $H$ that maps keys to hash values 
in a large range of size $q$ (where $q$ will be chosen later to bound the 
probability of error).  
The IBLT uses a hash function $H$ that maps keys to hash values 
in a large range of size $q$ (where $q$ is a power of 2 that will be chosen later to bound the probability of error).  
For the purpose of this paper we assume that $H$ is a
fully random hash function\footnote{We note that it is possible to show that
$H(x) = (a^x \bmod p) \bmod q$, where $1<a<p-1$ is a random positive integer and $p>q^2$,
has the needed properties when all keys are smaller than $q$.}.
Each key $x$ is placed into cells $T[h_1(x)]$, $T[h_2(x)], \dots, T[h_k(x)]$,
where again $T$ is the lookup table that represents the IBLT, as follows.  

Each cell $T[i]$ contains an ordered pair 
%
$$(\textsf{keySum},\textsf{keyhashSum}) \enspace .$$
The \textsf{keySum} field is the XOR of all the keys that have
been mapped to this cell, and hence must be the size of the keys (in bits).
The \textsf{keyhashSum} field is
the XOR of all the hash values $H(x)$ that have been mapped to this
cell, and hence must be the size of the hash $H$ of the keys (in bits).  
Note that insertion and deletion is the same operation, as
a deletion operation reverses an insertion.  Hence it is possible
to delete a key without it first being inserted.

\paragraph{Set reconciliation using IBLTs} 
The above structure yields a set reconciliation algorithm.  Consider two parties
Angel and Buffy (or $A_1$ and $A_2$).  Angel places his keys into an
IBLT, as does Buffy.  They are assumed to share the hash
functions $h_i$ and $H$ according to some prior arrangement.  They transfer their corresponding IBLTs, and
each then deletes their own keys from the transferred IBLTs.  Each
IBLT then contains the set difference, and the set difference is
recovered using the listing process.  Alternatively, since deleting
and inserting both correspond to XOR operations, we can say that the
parties take the {\em sum} of the IBLTs, by which we simply mean that
for each field in each cell, the corresponding values are summed via
the bitwise XOR operation.  As long as the set difference size $d$ is less
than the threshold $t$, recovery will occur with high probability (in
$t$).


As it will help us subsequently, we describe how the listing process functions.
Listing the contents of IBLTs uses a ``peeling process''.  Peeling
here corresponds to finding a cell with exactly one key contained in it, after
the insertion/deletion steps by Angel and Buffy have effectively removed keys
that appear in both sets.  To find a cell with one key, we check 
the \textsf{keySum} field using \textsf{keyhashSum}.  That is, if 
the \textsf{keySum} field contains a value $z$, we check whether
\textsf{keyhashSum} contains the value $H(z)$.  If $z$ is the actual key
contained in the cell, then $H(z)$ will indeed appear in the 
\textsf{keyhashSum} field.  If the \textsf{keySum} field contains the XOR
of several keys, then (under our assumption of random hash values for
$H$) there will be a false positive with probability only $1/q$ where
$q$ is the size of the range of hash values for $H$.  Let us temporarily assume 
that there are no false positives.

Once we have a cell with a single recoverable key $z$, we can remove
$z$ from the structure by computing $h_i(z)$ for all $i$ and deleting
$z$ from the corresponding cells using exclusive-or operations to
update the \textsf{keySum} and \textsf{keyhashSum} fields.  
Removing a key from the structure may yield new keys that can be
recovered.  This peeling process has been used in a variety of
contexts, such as in erasure-corrected codes \cite{LMSS}.  The peeling
process may also fail simply because at some point there may not be
any available cell with only a single recoverable key.  It can be shown
that recovery occurs with high probability, assuming a suitably-sized
IBLT is used \cite{IBLT,LMSS,Molloy}.  Specifically, the process of peeling
corresponds to finding what is known as the 2-core -- the
maximal subgraph where all vertices are adjacent to at least two
hyperedges -- on a hypergraph where cells correspond to vertices and each key 
$x$ corresponds to the hyperedge $\{h_1(x),\dots,h_k(x)\}$.
When $k$ is a constant, the peeling process yields an empty $k$-core with high probability whenever
the table size $m$ satisfies $m > (c_k + \epsilon)t$ for a constant threshold 
coefficient $c_k$ and constant $\epsilon > 0$.  
As noted in \cite{IBLT}, the threshold coefficients, given in Table~\ref{tab:thr}, are close to 1.
(Again, they can be made closer to 1 if desired using irregular hypergraph constructions \cite{IBLT,Rink}.)

\begin{table}[ht]
\begin{center}
\begin{tabular}{c|ccccc}
$k$ & 3 & 4 & 5 & 6 & 7\\\hline
$c_k$ & 1.222 & 1.295 & 1.425  & 1.570  & 1.721 \\
\end{tabular}
\end{center}
\caption{Threshold coefficients for the 2-core rounded to four decimal places.}
\label{tab:thr}
\end{table}

The following theorem, paraphrased from \cite{IBLT}, provides the probabilistic bounds on the failure
probability of the peeling process.
\begin{theorem}
\label{thm:thm1}
As long as we choose $m > (c_k + \epsilon) t$ for some
$\epsilon > 0$, the listing operation (not counting the separate probability of false positives from \textsf{keyhashSum}) fails with
probability $O(t^{-k+2})$.
\end{theorem}

We now bound the running time for peeling and error probability from false positives
from the \textsf{keyhashSum} field.  Given an IBLT to peel, we can start
by taking a pass over the $O(t)$ cells of the IBLT to find cells where the 
\textsf{keySum} field contains a value $z$ and the \textsf{keyhashSum} contains the value $H(z)$.
We keep a list of such cells and start the peeling with these cells.
As we peel a cell we update the \textsf{keySum}
and \textsf{keyhashSum} fields of other cells.  As we proceed through the list, we might
encounter cells that have already been peeled --- these are simply ignored.  
Also, while peeling
we test cells that we delete keys from to see if now the \textsf{keySum}
and \textsf{keyhashSum} fields match, in which case the cell can be added
to the list of cells to peel.  Overall this process clearly takes $O(t)$ time
(constant time per peeling operation), and there are $O(t)$ times that
we compare \textsf{keySum} and \textsf{keyhashSum} fields within a cell,
each of which can yield a false positive with probability $1/q$.  
Under a worst-case assumption that any such error would cause a
listing failure, this gives a total error probability of at most $O(t/q)$
caused by a false positive in the \textsf{keyhashSum} field.
We can choose $q$ according to our desired error bound.

\paragraph{Adapting to the set difference size}
If we have an upper bound on $d$, we can use this upper bound as the
value $t$ for the IBLT, and apply Theorem~\ref{thm:thm1}.  We
henceforth assume that an upper bound within a (small) constant 
factor of $d$ is available throughout this
work, as finding an upper bound is essentially an orthogonal problem.
Without an upper bound on $d$, some additional work may need to be
done, as explored in for \cite{setdifference,minsky2003set}; we
summarize these results and offer other alternatives as they apply
here.  Both of \cite{setdifference,minsky2003set} suggest approaches that
correspond to repeated doubling; if the IBLT size is not sufficient,
so that decoding is unsuccessful, then start over with larger IBLTs.
Another option is to double the IBLT size by
adding one additional lower-order bit to each hash value. Then it suffices
to send every odd-numbered cell of the IBLT arrays, since the even-numbered
cells can be found by subtraction from the old IBLT. In this way, the total
number of cells transmitted is the same as if the final IBLT had been
transmitted initially.
If the set difference is small but still a constant fraction of the
union of the set sizes, then using min-wise independence or related
techniques \cite{minwise,cohen} to approximate the set difference may
be suitable.  Finally, one might incrementally improve the IBLT.  If
each hash function is assigned its own subarea of cells (as is often
how Bloom filters are implemented to allow parallel lookups into the
structure \cite{BM}, and as noted in \cite{IBLT} is advantageous for
IBLTs), then the parties can incrementally add another hash function
by sending additional information corresponding to the subarea for 
the additional hash function. 

\paragraph{Keeping count}
Finally, the IBLT can also contain an optional \textsf{count} field,
which gives a count for the number of keys in a cell.  We can increase
the count by 1 on insertion, and decrease it by 1 on deletion.  With a
count field, a cell can contain a recoverable key only if the count is
1 or $-1$.  (From Angel's point of view, after deleting his keys from
Buffy's IBLT, when the \textsf{count} is $-1$, it could correspond to
a cell containing a key of his that Buffy does not have.)  Note,
however, a \textsf{count} of 1 or $-1$ does not necessarily correspond
to a cell with a recoverable key.  For example, if Angel has two keys
Buffy does not hold that hash to the same cell, and Buffy has one key
that Angel does not have that hashes to the same cell, then after
taking the difference of the IBLTs the \textsf{count} will be 1 but
there will be as sum of three keys in the cell.  The \textsf{count}
field can be useful for implementation and assists by acting like the
sum of another trivial ``hash'' value for the key (all keys
hash to 1), but it does not replace the \textsf{keyhashSum} field.

\paragraph{Abstraction}
At an abstract level, we can view the 2-party setting as follows.  
We desire a linear sketch (over an appropriate field) taken on sets 
of keys with the following properties.  Let $f(S)$ denote the sketch
of the set $S$.  We desire
\begin{itemize}
\item $f((S_1 \cup S_2) - (S_1 \cap S_2)) = f(S_1) \oplus f(S_2)$;
\item  a set $X$ can be efficiently extracted from 
$f(X)$ under suitable conditions, which here means that $X$ is sufficiently
small;
\item the size of the sketch is small, which here means that that if we want to 
recover $(S_1 \cup S_2) - (S_1 \cap S_2)$ then the sketch size is $O(|(S_1 \cup S_2) - (S_1 \cap S_2)|)$.
\end{itemize}
We have focused our description on IBLTs as it is a sketch with the required properties.
For multi-party reconciliation, we now show that IBLTs can be extended
by working over an appropriate field to ensure that with multiple sets $S_1,S_2,\ldots,S_n$,
we have $f(\cup_i S_i - \cap_i S_i) = \sum_i f(S)$, while still maintaining suitable extraction
and size properties.  

\section{The 3-Party Setting}

We now describe the extension of IBLTs to provide set reconciliation for three parties.
Starting with the 3-party setting allows us to demonstrate the key ideas behind this
approach and explore its capabilities;  we then examine how these extensions can be 
used beyond three parties.  

To start each of our three parties -- Angel, Buffy, and Cordelia (or $A_1$,
$A_2$, and $A_3$) -- inserts keys and hashes of keys
into the IBLT.  (For now, we will not use a \textsf{count} field.)  However, in
this setting both the keys and the hashes of keys are mapped injectively to
values in $(\mathbb{F}_3)^b$ for an appropriate $b$.  
The particular way of mapping to $(\mathbb{F}_3)^b$ does not matter --- we could interpret the key and hash values as number base 3 (at the cost of converting to base 3),
or (at the cost of some space) we could interpret the vector of bits of the key or hash value as a vector in $(\mathbb{F}_3)^b$.
Now, instead of using XOR in our insertion
operation -- which is equivalent to treating keys as $b$-bit 
elements of $(\mathbb{F}_2)^b$ -- we move to $(\mathbb{F}_3)^b$ (treating keys
as sequences of $b$ trits;  similarly, hash values are sequences of trits, perhaps of
a different length).  The three IBLT data structures are combined by
summing each of the \textsf{keySum} fields for each cell, as well as
summing the \textsf{keyhashSum} fields for each cell, where the sums
are sums of elements in $(\mathbb{F}_3)^b$.  To begin, let us ignore the issues of the
underlying network and assume that all parties obtain all 3 IBLTs.

Any key that appears in all three sets is canceled out
of \textsf{keySum} by the summation, and similarly the summation of
the three hashes of the key is canceled out of
\textsf{keyhashSum}.  Hence the number of keys existing in the IBLT
after this cancellation is $|\cup_i S_i - \cap_i S_i|$.  
If a key $x$ is found in the \textsf{keySum}
field and a matching $H(x)$ is in the \textsf{keyhashSum} field,
the key is recovered and removed by subtracting $x$ and $H(x)$
(or equivalently, adding in $2x$ and $2H(x)$ in the appropriate fields).
However, some
of these keys may appear duplicated in the IBLT; for example, if Angel
and Buffy have a key $x$ but Cordelia does not, then the sum of the 3
IBLTs may have a cell containing the value $2x$ in the \textsf{keySum}
field and $2H(x)$ in the \textsf{keyhashSum} field.  We therefore must
further modify our method of recovery.  If we see a value $z$ in the \textsf{keySum} field,
we must check whether the value $H(z)$ appears in the \textsf{keyhashSum}
field, but we must also check whether $2H(z/2)$ appears in the \textsf{keyhashSum}
field, in which case we treat it is a verification of $z/2$ as the key.  
This increases our error rate due to false positives from
\textsf{keyhashSum} by at most a factor of two, which is still $O(t/q)$.  
We note that we could reduce this error rate by instead keeping a count field,
which would tell us which one of the two cases above may apply.
Also, we note that when removing this key from the IBLT, each cell it is hashed
to would then have to remove two copies of the key.  

\begin{figure*}[t]
\begin{minipage}[t]{0.2\linewidth}
\centering
\includegraphics[scale=0.3]{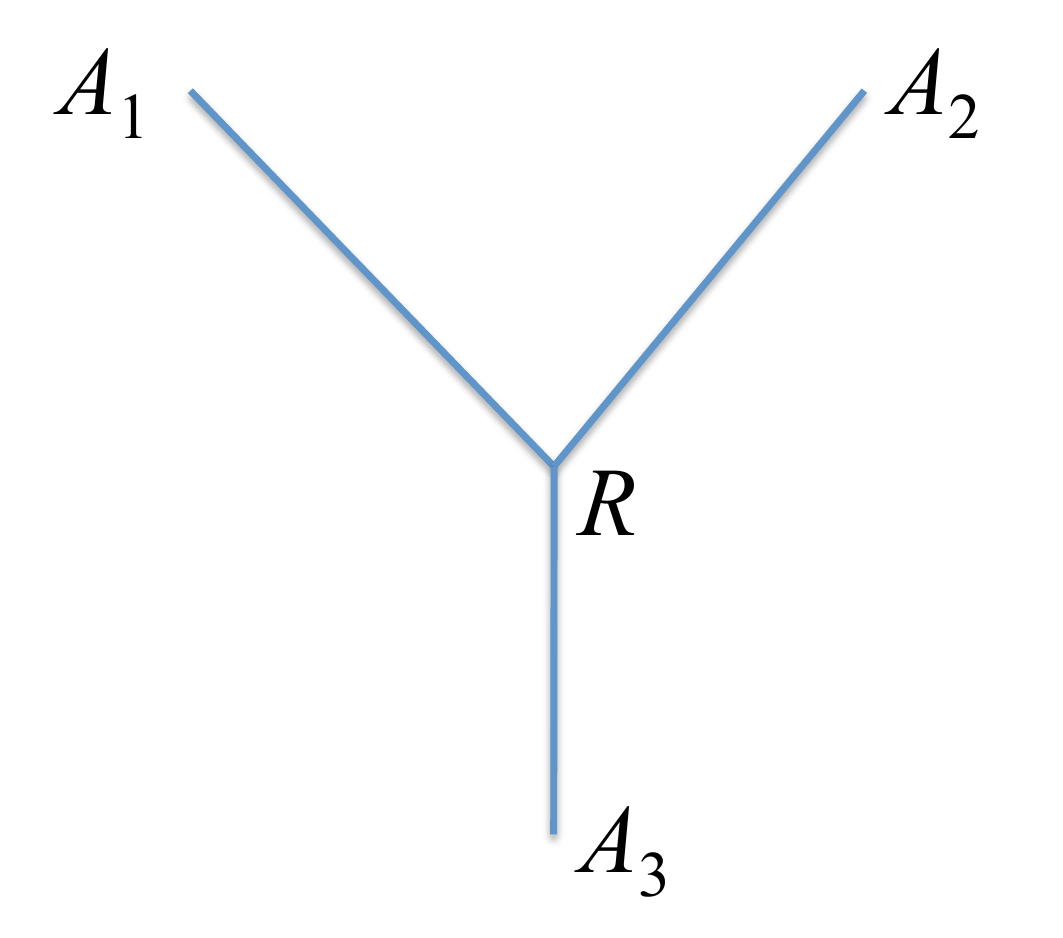}
\end{minipage}
\hspace{0.15cm}
\begin{minipage}[t]{0.2\linewidth}
\centering
\includegraphics[scale=0.3]{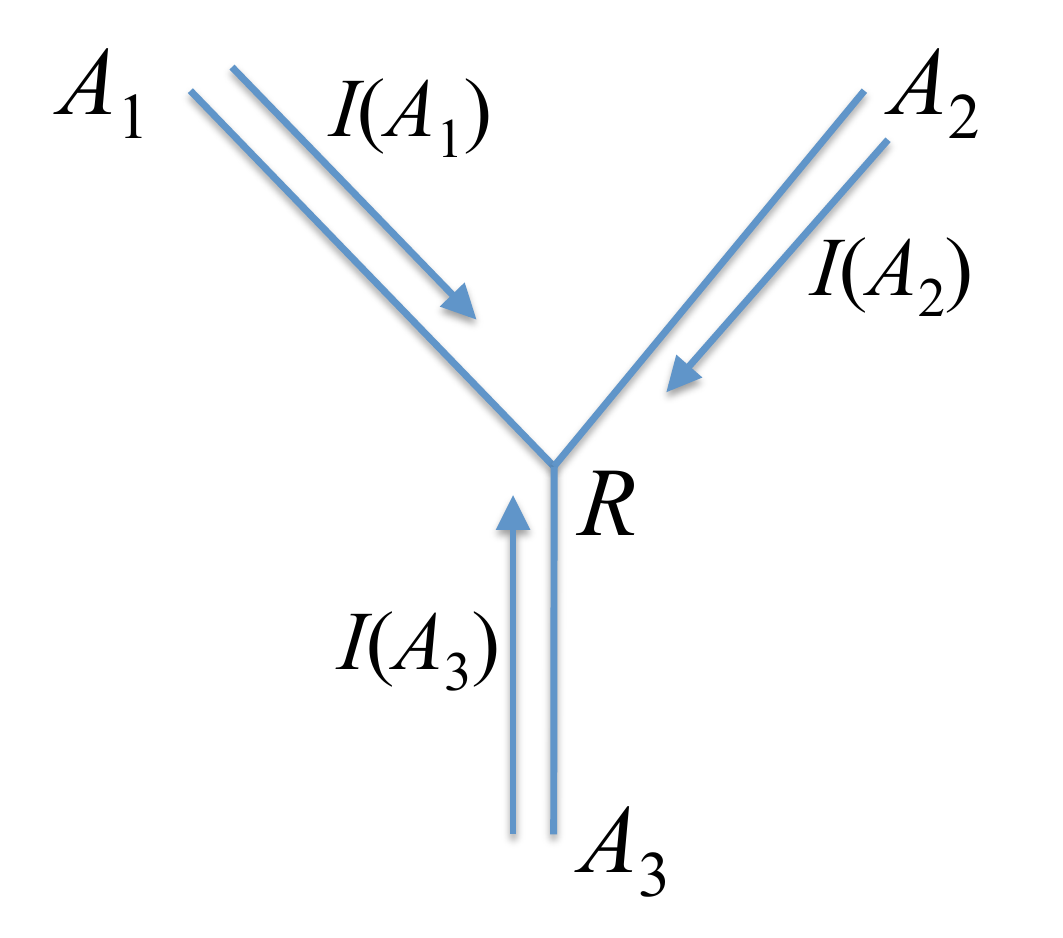}
\end{minipage}
\hspace{0.15cm}
\begin{minipage}[t]{0.2\linewidth}
\centering
\includegraphics[scale=0.3]{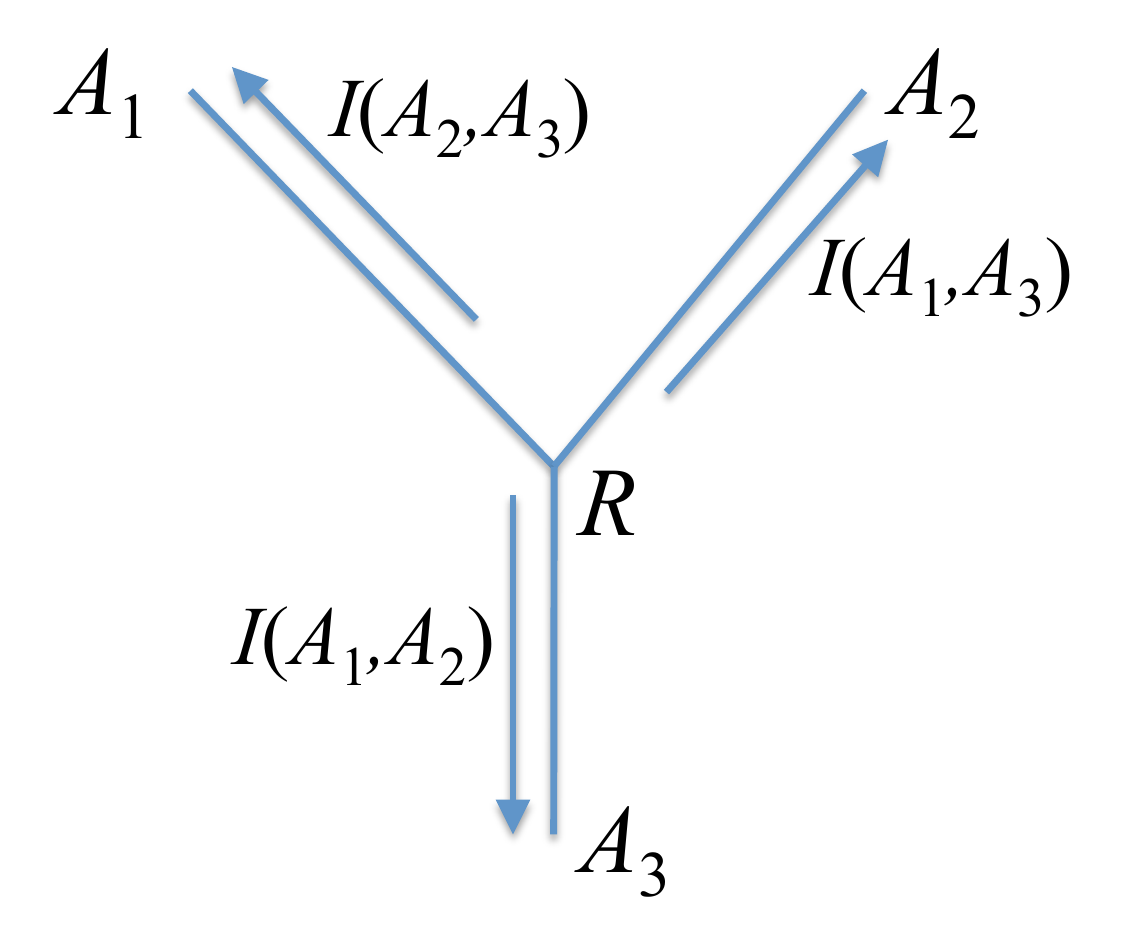}
\end{minipage}
\hspace{0.15cm}
\begin{minipage}[t]{0.2\linewidth}
\centering
\includegraphics[scale=0.3]{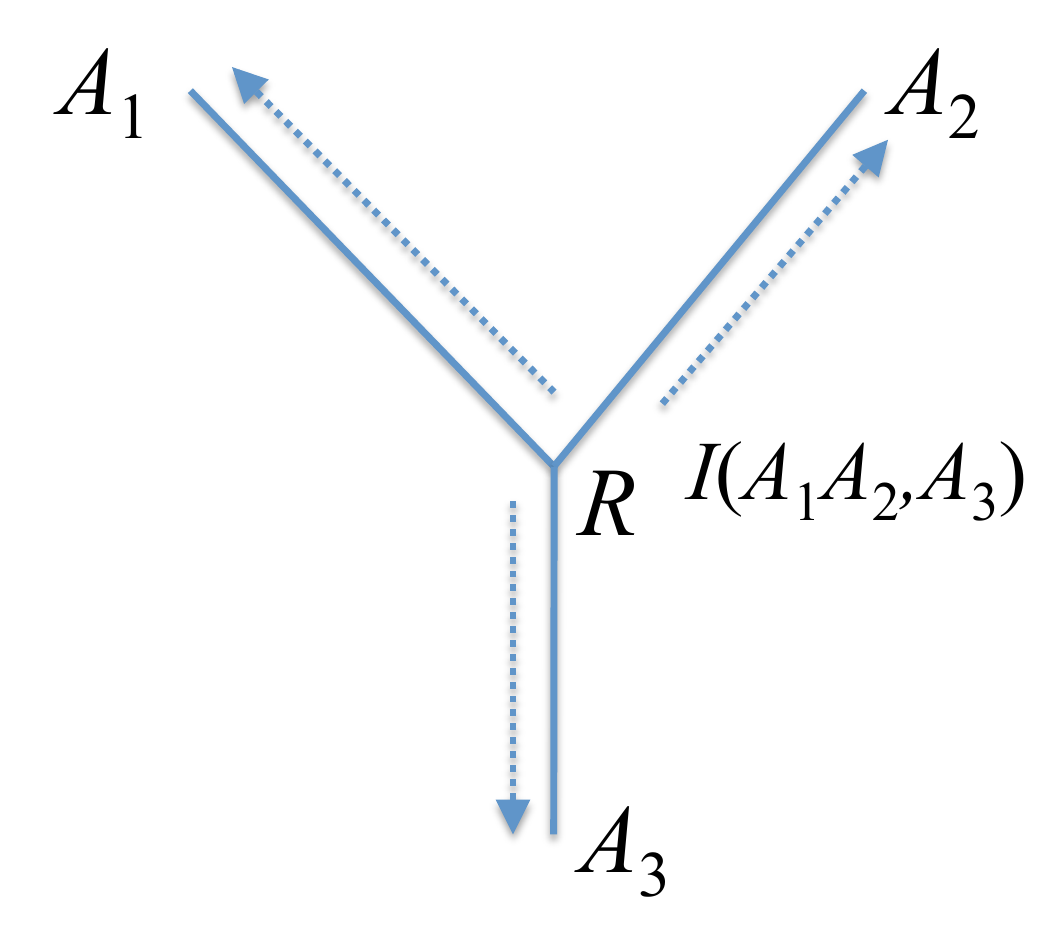}
\end{minipage}
\caption{Set reconciliation with network coding. The relay $R$ receives an IBLT from each party, and sends back to each party the sum of the other parties' IBLTs. Each party just needs to send and receive one message, improving on pairwise set reconciliation protocols.  Alternatively, in the subfigure on the right, in a wireless setting the relay $R$ can broadcast the sum of all the IBLTs to all parties (denoted by dotted lines), reducing the number of messages further.}
\label{fig:netc}
\end{figure*}

We emphasize that despite this difference, the IBLT listing process works in the same 
manner, and in particular has the same threshold size for successful listing.  This
is because a key is recovered exactly when a key is the only key hashed to a cell;  the
multiplicity of that key within the cell does not affect the listing process.
As the IBLT recovery works in the same manner as in the 2-party case,
we have the following theorem, based on Theorem~\ref{thm:thm1}.  
\begin{theorem}
\label{thm:thm2}
Consider a 3-party reconciliation using IBLTs with $k$ hash functions and a range
of $q$ values in the \textsf{keyhashSum} field.  
As long as we choose $m > (c_k + \epsilon) t$ for some
$\epsilon > 0$, and $t \geq |\cup_i S_i - \cap_i S_i|$, the 3-way reconciliation protocol fails with 
probability $O(t^{-k+2} + t/q)$.
\end{theorem}
The 3-party protocol uses $m$ cells, which contain the key and a hash.
As long as the keys are sufficiently large so that the hash is relatively small
compared to the keys, the constant factor from overhead will be small.  For
64-bit keys and 32-bit hashes, for example, the total overhead should be less
than a factor of 2 for $k =3,4$, and less than a factor of 3 for $k$ up to 7.    
In many settings, keys or associated stored values can be significantly longer than 64
bits and the overhead will be much smaller.  

\subsection{Useful Extensions}
\label{sec:usefulext}

Note that from this process each party can determine the number of other
parties (1 or 2) that hold a key they do not possess.  It is not hard
to add some additional information so that each party can determine 
which other party holds the key.  For example, we could add a 3-bit 
\textsf{IDs} field to each cell;  each time $A_i$ adds (or removes) a
key from the IBLT it would toggle the $i$th bit.  Hence the $i$th
bit of the \textsf{IDs} field would record the parity of the number of
keys that $A_i$ has added to the cell.  Recall that when we recover a key
from the cell, it should be the only key in that cell;  hence, the bits
set to 1 in the \textsf{IDs} field correspond to the parties that hold
that key.  (We note that, in fact, having a modulo 3 counter, and having
$A_i$ add $i$ to that counter when adding a key, would in fact suffice
in this case;  the details are left to the reader.  Our description here generalizes more readily.)  

Another reasonable question one might ask is if one party drops out of the
protocol, can the remaining two parties still reconcile their sets.
The answer is yes.  Suppose Cordelia does not participate, so that
only Angel and Buffy swap IBLTs.  In this case, when combining IBLTs,
Angel adds his own IBLT twice (or simply multiples every entry in his
IBLT by 2 initially), and similarly for Buffy.  That is, each participating party
can simply act as though they were two parties with the same set.
This guarantees that any key shared by both parties appears 3 times in
the IBLT and is canceled; the IBLT listing can be done as above.
Note that a participating party must know the number of other participating parties, 
potentially requiring dropping parties to signal their dropping out in some way.  

\subsection{Combining with Network Coding}

We have thus far assumed that all participating parties get all IBLTs,
and hence it may not be clear that this approach is significantly
advantageous when compared to simply performing pairwise
reconciliations.  However, significant advantages become clearer when
we consider the transmission of IBLTs over a network.  Because IBLTs
are linear sketches, based solely on addition, linear network coding
methods can be applied.

Specifically, suppose $A_1$, $A_2$, and $A_3$ are communicating over a
network via a relay $R$.   We emphasize that this is a simple setting 
for illustrative purposes.

In Figure~\ref{fig:netc} we let $I(A_i)$
represent the IBLT for $A_i$, and similarly let $I(A_i,A_j)$ be the
sum of the IBLTs for $A_i$ and $A_j$, and so on.  If we sent $I(A_1)$ to $A_2$
and $A_3$, even if the relay duplicates the IBLT it will have to cross
3 links, and similarly for the other two IBLTs.  Hence, the total
transmission cost will be 9 IBLTs worth of data.  However, suppose as
shown in Figure~\ref{fig:netc} that all the Bloom filters $I(A_1),I(A_2)$, and
$I(A_3)$ are sent to the relay $R$, and $R$ then takes sums to send
$I(A_2,A_3)$ to $A_1$, and similarly for the other parties.  Now only
6 IBLTs worth of data need to be sent, saving 1/3 of the transmission
cost.  This savings is entirely similar to standard network coding 
techniques.

Indeed, in the wireless setting, we could instead have the relay $R$
broadcast the joint IBLT $I(A_1,A_2,A_3)$ of all parties to all the
parties, reducing the number of messages down to four.  This approach
is similar to the now well-known approach of using simple XOR-based
network coding in wireless networks \cite{xors}.

We note that, in this relay setting, instead of using an IBLT with
entries over $(\mathbb{F}_3)^b$, we could build up a joint IBLT using
standard IBLTs by having the relay do more work.  Given IBLTs $I(A_1)$
and $I(A_2)$, the relay could determine the
set difference from the pair of IBLTs, and correspondingly add elements to
one IBLT to create an IBLT for the union of the sets.  Then on the arrival of $I(A_3)$
the relay could again determine the set difference between $A_1 \cup A_2$
and $A_3$ and use this to build an IBLT for $A_1 \cup A_2 \cup A_3$.  This
repeated decoding and encoding approach is used in \cite{BoralM}.  However, this approach requires much
more work from the relay, namely a full decoding for each newly received
IBLT, which should be avoided in many settings.  Our work shows, for the first
time, that such additional work can naturally be avoided.  

\section{Generalizing to $n$ Parties}

We now consider the generalization to $n$ parties.   We use a field of characteristic $p$, where $p \geq
n$, and we assume keys are mapped injectively (in an arbitrary way) into $(\mathbb{F}_p)^b$, for some $b$.
For convenience we simply take $p$ to be a prime here, so that
keys are mapped to vectors of non-negative integers smaller than $p$, and sums are computed modulo $p$.  
(And similarly for hash values, for a possibly different $b$).  
For simplicity, the reader might think about $b=1$, so keys and hash values are mapped
injectively to integers modulo $p$.  
For convenience we will assume multiplication and division modulo $p$ can be done in constant time;
those who object to this assumption may add an appropriate $O(\log^2 p)$ factor to the time bounds
(although better bounds may be possible with specially chosen primes).
Let $v_i$ denote the representation of $I(A_i)$ as
a vector as above.


For efficiency and to reduce the probability
of a false positive when using the
\textsf{keyhashSum} to verify the key value in the cell, we keep a 
counter modulo $p$ in the \textsf{count} field to track the count of
the number of (copies of) keys hashed to a cell by all parties

We first state a general result that may not appear directly relevant at first blush.
However, this form is useful in that it can be applied to more specific situations,
including not only the straightforward generalization to $n$ parties, where we consider
a sum of $n$ IBLTs, but also situations (motivated by randomized network coding) in which
we are given two different linear combinations of IBLTs, and need to do set reconciliation.
The condition under which set reconciliation succeeds is somewhat technical, and it is possible
that some set $X$ of keys cannot be recovered. But as we will see, for the linear combinations of interest (such as random linear combinations) the probability that $X\ne\emptyset$ can be made small;  in some settings (e.g., when a sum of IBLTs is obtained as in from a relay) the probability will be 0.  

\begin{theorem}
\label{thm:mostgen}
Consider an $n$-party reconciliation using IBLTs with $k$ hash functions and a range
of $q$ values in the \textsf{keyhashSum} field.  
Suppose we know two linear combinations (over $(\mathbb{F}_p)^b$)
$L_1 = \sum_{i=1}^n \alpha_i v_i$ and $L_2 = \sum_{i=1}^n \beta_i v_i$, as well as the sums of coefficients
$\alpha = \sum_i \alpha_i \bmod p$ and $\beta = \sum_i \beta_i \bmod p$, where $\alpha_i \neq 0 \bmod p$ for all $i$,
and $\beta \neq 0 \bmod p$.
Let $X= \{ x\in \cup_i S_i \; | \; \sum_i (\alpha_i + \gamma \beta_i) 1_{x\in S_i} \bmod p = 0\}$, where $\gamma = -\alpha \beta^{-1} \bmod p$, and $1_E$ is an indicator random variable for the event $E$.
As long as we choose $m > (c_k + \epsilon) t$ for some
$\epsilon > 0$, and $t \geq |\cup_i S_i - \cap_i S_i|$, then from $L_1$ and $L_2$ we can determine 
all keys from $(\cup_i S_i - \cap_i S_i) - X$, 
with probability $1-O(t^{-k+2} + t/q)$.
\end{theorem}
\begin{proof}
If a key is present in all sets $S_i$ then it appears with coefficient $\alpha$ in $L_1$, and similarly it will appear
with coefficient $\beta$ in $L_2$.  From these we can form the combination of $L_1 + \gamma L_2$, where as stated $\gamma = -\alpha \beta^{-1} \bmod p$.  The coefficient of a key that is present in all sets $S_i$ is then 0 modulo $p$ and therefore the key is not present in the IBLT $L_1 + \gamma L_2$.  Unfortunately, the same is true for any key that appears exactly in sets $S_j$ for $j \in T$ when $T$ has the property that $\sum_{i \in T} (\alpha_i + \gamma \beta_i) = 0 \bmod p$.  All other keys can be found with the given probability using the IBLT recovery process.  Note we assume the use of a \textsf{count} field that tracks the weighted multiplicity of the number of keys hashed to a cell modulo $p$ (that is, the sum of the coefficients of the keys hashed to that cell), so that only a single possible key value must be tested in a cell at any time.  It is this use of the \textsf{count} field that limits the additional failure probability to $O(t/q)$.    

We also remark that in the case where $\alpha = 0$,
the theorem also holds, but in fact there is no need for the second
linear combination $L_2$ (as $\gamma = 0$).  In this case, a key present in all sets has
a multiplicity that is $0 \bmod p$, and $X$ corresponds those keys
$x$ that appear exactly in sets $S_j$ for $j \in T$ where $T$ has the
property that $\sum_{i \in T} \alpha_i = 0 \bmod p$.  
{\hspace*{\fill}\rule{6pt}{6pt}\smallskip}
\end{proof}

To see how Theorem~\ref{thm:mostgen} can be used, we first consider
the basic case, where each party simply wants to find $\cup_i S_i
- \cap_i S_i$.  We first provide the argument without using the theorem,
and then see how the theorem immediately implies the result.  
Let us assume that each party $A_i$ obtains the sum of all IBLTs $Z = \sum_i v_i$, and
of course each party $A_i$ also has $S_i$ and hence $v_i$.  
If $p > n$, then after $A_i$ obtains the sum $Z$ of all IBLTs, it computes the
sum $Z+(p-n)v_i$, essentially acting as $p-n+1$ parties with the same set. 
As before, this means that the contributions of keys appearing in all sets will
cancel out.  Each key not appearing in all sets will have an associated multiplicity
of less than $n$ in $Z$ and hence less than $p$ in $Z+(p-n)v_i$ (regardless of whether
the key is in $S_i$ or not). The algorithm can then for each cell examine the \textsf{count}
$a$, the \textsf{keySum}~$y$, and the \textsf{keyhashSum} $z$ to determine if $H(a^{-1}y)=a^{-1} z \bmod p$, which is true when the count $a$ corresponds to a single key.
Hence, as before, each key not in all sets can be recovered using the IBLT recovery process
with probability $1-O(t^{-k+2} + t/q)$.
(For small $n$ one might choose not to use a \textsf{count} field;
one could avoid keeping the \textsf{count} field and test all possible \textsf{count} values,
that is try all values of $a$ from 1 to $p-1$. Or one can be
slightly smarter;  if $A_i$ contains a single $x \in S_i$ that hashes to that cell,
then $A_i$ need only test the value of 
$a$ that satisfies $ax = y$, and 
if $A_i$ has $x \in S_i$ that hashes to that cell, then only $a$ values from
$1$ to $n-1$ need to be tested.)  

Alternatively, we see that this matches the setting of
Theorem~\ref{thm:mostgen}, with $L_1 = Z$,  $L_2 = v_i$, and $\gamma = -n$.  
Computing $Z-nv_i = Z+(p-n)v_i \bmod p$ we see that the set $X$ must be empty,
because as we argued above any key not in all sets has a count that is non-zero 
modulo $p$ in $Z+(p-n)v_i$.  The result follows.  

\begin{figure*}[t]
\begin{minipage}[t]{0.95\linewidth}
\centering
\includegraphics[scale=0.6]{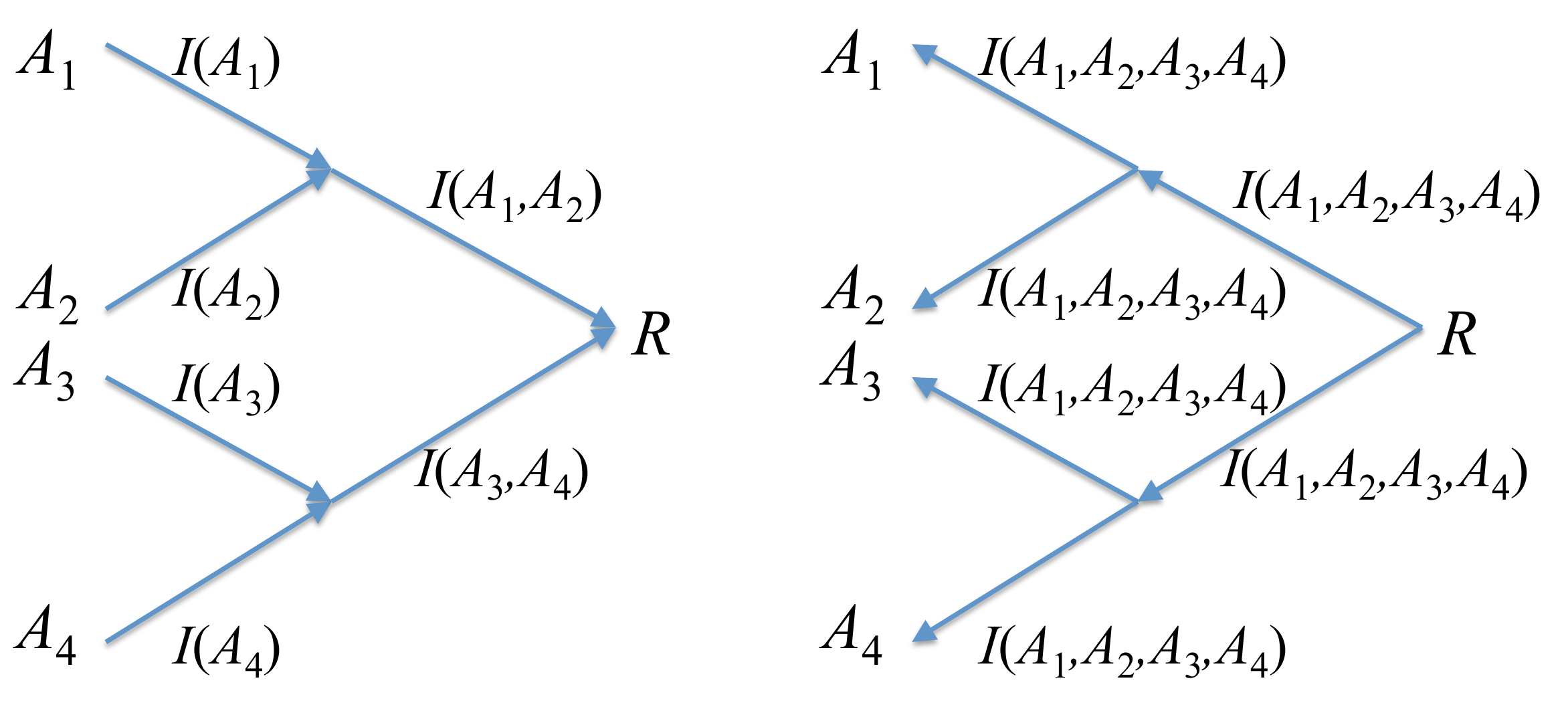}
\end{minipage}
\caption{Set reconciliation example when the parties are connected by a communication tree, a common case in network
communications.}
\label{fig:tree}
\end{figure*}

\subsection{Extensions}

Several of our extensions from Section~\ref{sec:usefulext} hold in this framework as well.
For example, with an $n$-bit 
\textsf{IDs} field one can track the parity of the number of keys
in a cell for each of the $n$ parties, and thereby determine the parties that hold a recovered key.
Also, if some parties do not
participate, each participating party can simply add in additional
copies of its own IBLT to arrange for cancellation if all
participating parties hold a key.  All that needs to be known for
recovery is the number of participating parties.

In the case of $n$ parties connected through a relay node, we find
that by passing IBLTs through the relay, we can arrange for $n$-party
set reconciliation using $2n$ messages on a wired network and $n+1$
messages on a wireless network, where the relay broadcasts the sum of
the IBLTs to all parties.  This improves over the simplistic natural approach
of using $n(n-1)$ point-to-point messages to compute all pairwise differences.
However, it should be noted, the pairwise difference messages may in fact
be smaller in size, since pairwise set differences may be smaller than the total
set difference.

\section{Network Coding and IBLTs}

At a high level, our work thus far suggests 
that, by working over a
suitable finite field, IBLTs can be naturally plugged in to linear
network coding schemes to provide efficient set reconciliation
mechanisms, where the message size corresponds (up to constant
factors) to the generalized set difference.  We believe this correspondence
is indeed quite general, and while the broad nature of use and applications
of network coding make it difficult to turn this statement into a theorem,
we provide some sample applications.  

\subsection{Set Reconciliation on Trees}

We first consider set reconciliation among $n$ parties, connected by a
communication network that is a {\em rooted tree} with $E$ edges,
known in advance.  The parties are at the leaves of the tree.  At each
time step, each node in the tree can send a message to each of its
neighbors.  Let $P$ be the length of the longest path from a leaf to a root of the tree.  
We assume in what
follows that keys are sufficiently large so that other overhead (e.g.,
the \textsf{keyhashSum} field) at most affects the total size by a
constant factor, and that an upper bound on $|\cup_i S_i - \cap_i S_i|$
within a constant factor is known.  We claim the following.
\begin{theorem}
Set reconciliation among $n$ parties on a rooted tree can be performed 
in $2P$ time steps with $2E$ total messages of size $O(|\cup_i S_i
- \cap_i S_i|)$ with probability $1-O(t^{-k+2} + nt/q)$.  Each non-leaf node in the network requires
only $O(|\cup_i S_i - \cap_i S_i|)$ storage.  
\end{theorem}
\begin{proof}
Each party constructs its IBLT, and sends it up the communication
network to the root.  Non-leaf nodes excluding the root combine all
the IBLTs from leaves in their subtree, and pass them up to the root.
The root gathers all the IBLTs and sends the union of them back down
the communication tree.  (See Figure~\ref{fig:tree}.)  Messages are of size $O(|\cup_i S_i - \cap_i S_i|)$,
and each edge of the tree carries a single message in each direction.
The IBLTs allow each party to recover all keys with probability $1-O(t^{-k+2} + t/q)$.
This follows immediately from Theorem~\ref{thm:mostgen},
as again we are in the setting where $L_1 = \sum_i v_i$ and $L_2 = v_i$.
All of the IBLTs will be peeling the same set of keys, so the $O(t^{-k+2})$ term
in the failure probability is common to all $n$ parties.  The $n$ parties may have
different counts associated with different cells, however, from each adding in 
their own term $(p-n)v_i$, so we take a union bound
over the $O(t/q)$ failure probability associated with a false match from 
the \textsf{keyhashSum} field over the $n$ parties.  
(We note that if $n = p$, all $n$ parties can use the same IBLT, and this union
bound is avoided.)  Because of the way IBLTs are combined only 
$O(|\cup_i S_i - \cap_i S_i|)$ space is required at non-leaf nodes.
{\hspace*{\fill}\rule{6pt}{6pt}\smallskip}
\end{proof}

\subsection{Set Reconciliation via Gossip on General Networks}

We now show how gossip spreading techniques (also
referred to generally as rumor spreading) allow multi-party set
reconciliation over a network in $O(\phi^{-1} \log n)$ rounds of
communication, where $\phi$ is the conductance of the network, using
our IBLT framework.

In gossip spreading, there are generally two different models \cite{shah}.  In the
single-message model there is one message at a vertex in a graph with
$n$ vertices\footnote{For this problem we follow the standard notation
for gossip problems and use $n$ for the number of vertices.}, and the
goal is for every vertex to obtain that message.  In the multi-message
model, the standard setting is that for some subset of the vertices, each
vertex has a unique message, and all vertices have to obtain all of the
messages.  With a {\sc PUSH} strategy, in each round every node that
has a message contacts a random neighbor, and forwards a single
message.  The {\sc PULL} strategy is similar, but each node without a
message contacts a random neighbor and obtains a message from them.  A
{\sc PUSH}-{\sc PULL} strategy combines both of these operations in every round.
(See, for example, \cite{chierichettialmost,giakkoupis}.)
In these models, a vertex can transfer a single message to another vertex
in each round.
In the multi-message model, one can use network coding by having a vertex
send a (usually random) linear combination of the messages it holds
at this time, so that messages take (essentially) the same space, but 
can offer potentially more useful information.    (See, for example, \cite{DebMedard,haeupler}.)

Our setting does not exactly match any of these situations.  For
reconciliation, we have multiple vertices each with their own message
(the IBLT), and we can combine messages, so we appear to be similar
to the multi-message model with network coding.  However, in the
network coding setting, the eventual goal is to solve a collection of
linear equations (corresponding to the combinations of messages
received), and in that setting, each new message only provides one
more ``degree of freedom'', or one more needed equation, regardless of
the component messages it contains.  For set reconciliation, we need
not solve such a system (although that would be one way to solve the
problem); we merely need some appropriate linear combination of all of
the IBLTs, as demonstrated in Theorem~\ref{thm:mostgen}.  Also, as we
are working in the reconciliation setting, we do not require that all 
vertices obtain the reconciled sets, but only those vertices that begin
with a set initially.

Because of this, we may more naturally think of the problem as a
collection of single-message problems running in parallel, as we now
describe.  Our
approach is as follows.  We have $a \leq n$ parties
$A_1,\ldots,A_a$ with sets $S_1,\ldots,S_a$ and corresponding IBLTs $I(S_1),\ldots,I(S_a)$.
For convenience, without loss of generality we provide the argument where $a = n$, as fewer parties with
messages only makes things easier.  (Or we may think of parties without a message
as having a null message.)
Parties can use whatever single-message gossiping algorithm is available.  Messages
in this setting correspond to random mixtures of IBLTs.  That is, here we again
think of the IBLT as a vector of entries in $(\mathbb{F}_p)^b$ for a suitably large prime $p$.
(Here $p$ will need to be larger than previously to obtain a low failure probability, as we see below.)
A message will consist of a linear combination of such vectors
$$\sum_{i \in [n]} \alpha_i v_i,$$ where $v_i$ is the IBLT vector for
the $i$th party and $\alpha_i$ is a coefficient (modulo $p$).  Let the vector of the
$j$th party after $\ell$ rounds of communication be given by
$\alpha_{ij\ell}$, $i=1,\dots,n$.  Our goal is for each party $j$ to obtain a vector 
$\sum_{i \in [n]} \alpha_{ijL} v_j$ at
round $L$ where $\alpha_{ijL} \neq 0$ for all $i$.  At that point,
as a special case of Theorem~\ref{thm:mostgen},
the $j$th party can reconcile using the combined IBLT by adding
$(p-\sum_{i \in [n]} \alpha_{ijL}) v_j$ to this vector, thereby
``canceling out'' any key in the intersection of the sets.
To bound the probability that the set $X$ of non-recovered
keys is nonempty, we first need to describe how the coefficients come about.

The protocol runs as follows. 
Each vertex holds one linear combination of IBLTs at any time; after round number $\ell$ party $j$ stores
$$\sum_{i \in [n]} \alpha_{ij\ell} v_i$$
as well as the coefficient sum $\alpha_{j\ell} = \sum_{i \in [n]} \alpha_{ij\ell}$.
To send a message, a party $j$ chooses a random $\kappa_{j\ell} \neq 0$ modulo $p$
and sends
$$\kappa_{j\ell} \sum_{i \in [n]} \alpha_{ij\ell} v_i.$$
together with the corresponding coefficient sum $\kappa_{j\ell} \alpha_{j\ell} \bmod p$. 
Each party receiving a message simply adds it to its current message.
The probability of ever ``zeroing out'' a coefficient using this approach is negligible
for suitably large $p$.
To see this, notice that each coefficient $\alpha_{ij\ell}$ is a (non-zero) multilinear polynomial of degree at most $\ell$ of the random multipliers $\kappa_{j\ell}$ applied to each message. This implies that the probability that $\alpha_{ij\ell}=0$ is $O(\ell/p)$, by the Schwartz-Zippel lemma \cite{Schwartz,Zippel}. Similarly, the probability for each set $T$ that $\sum_{i\in T} \alpha_{ij\ell}$ assumes a given value is $O(\ell/p)$, implying that $X$ is empty with probability $1-O(t\ell/p)$ for each party.

Now we examine the protocol from the point of view of a single
message.  For a single message, the protocol behaves exactly as the
single-message protocol;  the fact that other IBLTs may be
piggy-backing along in a shared message does not make any difference.  
Hence, we can treat this as multiple single-message
problems running in parallel, and apply a union bound on the failure
probability.  Generally, standard results for single-message problems
come with an exponentially decreasing tail bound for the probability of not finishing
after a given number of rounds.  Assuming this, an additive $O(\log n)$ steps over a
standard single-message result are sufficient to guarantee, via union
bound, that the $n$ parallel problems all complete.

As a specific example, we can consider the best current results on the standard
{\sc PUSH}-{\sc PULL} protocol for gossip spreading \cite{giakkoupis}.  Let $\phi$ be the
conductance of a communications graph $G$ with $n$ vertices, where $a$
of the $n$ vertices wish to reconcile their sets.  We can prove the following
theorem.
\begin{theorem}
\label{thm:pushpull}
Set reconciliation on a graph of $n$ vertices with $a = O(n)$ parties
having sets to reconcile can be accomplished in time $O(\phi^{-1} \log
n)$ using the standard randomized {\sc PUSH}-{\sc PULL}
protocol with messages of size
$$O(|\cup_i S_i - \cap_i S_i|)$$
with success probability 
$$1-O(t^{-k+2} + nt/q+ntL/p+n^2L/p+n^{-\beta})$$
for any constant $\beta > 0$.
\end{theorem}
\begin{proof}
As mentioned without loss of generality we take the case where $a = n$.
We choose a suitable stopping time $L = O(\phi^{-1} \log n)$ based on the
choice of the constant $\beta$ that would be
suitable for $n$ parallel versions of the single-message {\sc PUSH}-{\sc PULL} gossip protocol to 
successfully complete with high probability, as guaranteed by Theorem 1.1
of \cite{giakkoupis}.  

We now apply Theorem~\ref{thm:mostgen}.  From our discussion above, we
have that after $L$ rounds the $i$th party will store $L_1$ that
is a linear combination $\sum_{i \in [n]} \alpha_{ijL} v_i$, as well
as $L_2 = v_i$.  Here the $\alpha_{ijL}$ are all non-zero with high
probability because of the use of random coefficients as discussed above; this
probability it at most $O(n^2 L/p)$ by the union bound as there are $n^2$
coefficients and each is 0 with probability at most $O(L/p)$.  

We further claim the set $X$ is empty for all parties with high probability, also
because of the use of random coefficients.  To see this, consider
any key $x$ 
in $\cup_i S_i - \cap_i S_i$.  For
$x$ to not be recoverable by the IBLT by the $i$th party,
it would require that $\sum_{j \in T} \alpha_{ijL} = 0 \bmod p$, where
$T$ is again the set of parties that have $x$ as a key.  This
happens with probability at most $O(tL/p)$ as discussed above.
Hence, by a union bound over parties, the probability that this
happens over all parties is at most $O(ntL/p)$.  

We therefore obtain full recovery for all parties using randomized {\sc PUSH}-{\sc PULL}, with
high probability.  Note that again we must take into account that each party has different
coefficients in their IBLT, and hence one must apply a union bound to cover
the possibility of false positives from the \textsf{keySum} and \textsf{keyhashSum} fields 
over all IBLTs.  However, the recovery process will be the same for all IBLTs, since they
all involve the same keys.  Our final probability bound includes this accounting.  
{\hspace*{\fill}\rule{6pt}{6pt}\smallskip}
\end{proof}

\subsection{Experiments}

We briefly describe some experiments designed to test the gossip algorithms approach.
We emphasize that the experiments were meant as ``proof-of-concept'', and not an extensive
experimental test.\footnote{These experiments were performed by Marco Gentili, who we thank
for allowing their use in this paper.}   

We choose as our test graphs random graphs where each edge is included
independently with probability $\frac{2 \ln n}{n}$; this is sufficient
to guarantee the graph is connected (with high probability).  For our
experiments, each each node is a party to the protocol, and each party's set is simply one element, with all sets
being distinct.  We use IBLTs of $2n$ cells (more than needed) to
ensure listing succeeds with high probability.  The choice of sets
does not significantly impact the failure probability.  Every time a
party receives a message, it adds the corresponding IBLT multiplied by
a random multiplier into its linear combination of IBLTs.  We use the
{\sc PUSH}-{\sc PULL} protocol as described previously.  We also
determined with preliminary experiments the number of rounds needed to
ensure that all parties would receive the information held from all
parties with high probability, and used this many rounds.  By doing so, we limit our failures to resulting from the
zeroing out of coefficients of the linear combination of IBLTs.  This
failure probability therefore corresponds to the $O(ntL/p + n^2L/p)$
term from Theorem~\ref{thm:pushpull}.

\begin{table}[!htb]
\caption{ Success Rate of Listing Keys (Averaged over 1000 trials) with $p=1000000007$}\label{table:p1000000007}
\centering
\begin{tabular}{|l|l|l|l|l|l|l|}
\hline
\begin{tabular}[c]{@{}l@{}}\# Parties\end{tabular} & \begin{tabular}[c]{@{}l@{}}\% \\Retrieving \\     All\\     Msgs\end{tabular} & \begin{tabular}[c]{@{}l@{}}\% \\Missing\\     1      Msg\end{tabular} & \begin{tabular}[c]{@{}l@{}}\% \\Missing\\     \textgreater1     Msg\end{tabular} & \begin{tabular}[c]{@{}l@{}}\# \\Retrieving \\     All     Msgs\end{tabular} & \begin{tabular}[c]{@{}l@{}}\# \\Missing\\     1      Msg\end{tabular} & \begin{tabular}[c]{@{}l@{}}\# \\Missing\\     \textgreater1     Msg\end{tabular} \\ \hline
10                                                                    & 100.00\%                                                                        & 0.00\%                                                                    & 0.00\%                                                                               & 10000                                                                           & 0                                                                         & 0                                                                                    \\ \hline
20                                                                    & 100.00\%                                                                         & 0.00\%                                                                    & 0.00\%                                                                               & 20000                                                                           & 0                                                                         & 0                                                                                    \\ \hline
40                                                                    & 100.00\%                                                                         & 0.00\%                                                                    & 0.00\%                                                                               & 40000                                                                           & 0                                                                         & 0                                                                                    \\ \hline
80                                                                    & 100.00\%                                                                         & 0.00\%                                                                    & 0.00\%                                                                               & 80000                                                                           & 0                                                                         & 0                                                                                    \\ \hline
160                                                                   & 100.00\%                                                                         & 0.00\%                                                                    & 0.00\%                                                                               & 160000                                                                          & 0                                                                        & 0                                                                                    \\ \hline
320                                                                   & 100.00\%                                                                         & 0.00\%                                                                    & 0.00\%                                                                               & 320000                                                                          & 0                                                                       & 0                                                                                    \\ \hline
640                                                                   & 100.00\%                                                                         & 0.00\%                                                                    & 0.00\%                                                                               & 640000                                                                          & 0                                                                       & 0                                                                                    \\ \hline
1280                                                                  & 100.00\%                                                                         & 0.00\%                                                                    & 0.00\%                                                                               & 1279999                                                                         & 1                                                                      & 0                                                                                    \\ \hline
\end{tabular}
\end{table}

Table \ref{table:p1000000007} shows the success rate for listing keys using $p = 1000000007$, averaged over 1000 trials.  All keys for $n$ up to 640 were reconciled; for $n=1280$, one key from one party is not recovered in one trial.  (Various backup measures could be used to easily handle such rare cases.)  Note that this
value of $p$ would fit into a 32-bit integer and is not unreasonable for calculations.  Other experiments
with smaller values of $p$ shows that the failures occur at a rate roughly inversely proportional to $p$, as suggested by Theorem~\ref{thm:pushpull}.   

\section{Conclusion}
Previous solutions for set reconciliation, based primarily on
Reed-Solomon codes, have not (as far as we are aware) been generalized
to the multi-party setting, and it is not immediately clear how to do
so.  We have shown here that methods based on Invertible Bloom Lookup
Tables, which require additional space but only linear time,
generalize naturally, and further they do so in a way that allows the
application of network-coding based techniques.  Hence, by utilizing
network coding methods, we can obtain high efficiency in terms of the
number of network messages required, as well as small messages because
IBLTs and combinations of IBLTs have length proportional to the
generalized set difference.  While we expect our approach might be improved
further, providing better space utilization or smaller probability of
error either by theoretical improvements or by careful implementation, we
believe this work represents an important step in establishing
more practical solutions to the multi-party set reconciliation problem than 
approaches based on pairwise interactions.

There are, of course, quite a number of linear sketches in the
literature beyond IBLTs, and combining such sketches is a fairly common
technique.  Our work emphasizes how IBLTs can naturally lead to
reconciliation algorithms that take advantage of methods based on
network coding.  It would be very interesting if a general statement
formalizing this connection more concretely could be developed, or,
alternatively, if we can find other cases where the utility of linear
sketches can be increased by applying network coding techniques to expand
their capabilities or efficiency.

We observe that in settings where point-to-point messages can be
transmitted more efficiently than by broadcasting, there is potential
in some cases to decrease the total size of messages sent and received
by each party. For example, this is the case when each set $S_i$ lacks
$t$ elements from $\cup_i S_i$, and each of these elements is present in all other
sets. Then running a single pairwise set reconciliation protocol with point-to-point messages of size $O(t)$ suffices for each
party. However, the IBLT would require each party to send and receive
messages of size $O(tn)$.  Of course, in this example the parties are making
use of the knowledge that all other parties already hold the missing items,
but it shows that there are settings where IBLTs will not be optimal.  
More generally, we leave it as an open problem to explore the possible 
trade-offs in using or combining various set reconciliation protocols in 
additional settings.

\section*{Acknowledgments}  The first author thanks George Varghese for suggesting 
the problem of multi-party set synchronization.  

\bibliographystyle{plain}

\end{document}